\documentclass{article}
\usepackage{color,xcolor}
\usepackage{amsmath}
\usepackage{amsfonts}
\usepackage{amssymb}
\usepackage{amsthm}
\usepackage{enumerate}
\usepackage{enumitem}
\usepackage{fullpage}
\usepackage{bm}
\usepackage{thmtools, thm-restate}
\usepackage[nottoc]{tocbibind}

\newtheorem{theorem}{Theorem}[section]
\newtheorem{lemma}[theorem]{Lemma}
\newtheorem{proposition}[theorem]{Proposition}
\newtheorem{claim}[theorem]{Claim}

\newtheorem*{theorem*}{Theorem \ref{thm:main}}

\newtheorem{definition}[theorem]{Definition}
\newtheorem{corollary}[theorem]{Corollary}
\newtheorem{example}[theorem]{Example}

\newcommand{\F}{\mathbb{F}}
\newcommand{\C}{\mathbb{C}}

\newcommand{\R}{\mathbb{R}}
\newcommand{\Z}{\mathbb{Z}}

\newcommand{\poly}{\mathsf{poly}}

\newcommand{\Fam}{\mathcal{F}}
\newcommand{\U}{\mathcal{U}}
\newcommand{\V}{\mathcal{V}}

\newcommand{\D}{\mathbf D}

\newcommand{\Otilde}{\Tilde{O}}

\newcommand{\Supp}{\mathsf{Supp}}
\newcommand{\Span}{\mathsf{Span}}

\def\anon{0}

\ifnum\anon=0
\newcommand{\fnote}[1]{{\color{brown} [Fatemeh: #1]}}
\newcommand{\snote}[1]{{\color{blue} [Swastik: #1]}}
\else  
\newcommand{\fnote}[1]{}
\newcommand{\snote}[1]{}
\fi 

\usepackage{parskip}

\title{Fourier Sparsity of Delta Functions\\and\\Matching Vector PIRs}
\author{Fatemeh Ghasemi\thanks{Department of Mathematics, University of Toronto, Canada. Email: \texttt{fatemeh.ghasemi@mail.utoronto.ca}}
\and 
Swastik Kopparty\thanks{Department of Mathematics and Department of Computer Science, University of Toronto, Canada. Research supported by an NSERC Discovery Grant.
Email: \texttt{swastik.kopparty@utoronto.ca}}
}
\date{}

\begin{document}
\maketitle

\begin{abstract}
In this paper we study a basic and natural question about Fourier analysis of Boolean functions, which has applications to the study of Matching Vector based Private Information Retrieval (PIR) schemes. 

For integers $m,r$, define a {\em delta function} on $\{0,1\}^r \subseteq \Z_m^r$ to be a function
$f: \Z_m^r  \to \mathbb C$
if $f(0) = 1$ and $f(x) = 0$ for all nonzero {\em Boolean} $x$.
The basic question that we study is how small can the Fourier sparsity of a delta function be; namely, how sparse can such an $f$ be in the Fourier basis?

In addition to being intrinsically interesting and natural, such questions arise naturally while studying ``$S$-decoding polynomials" for the known matching vector families. Finding $S$-decoding polynomials of reduced sparsity -- which corresponds to finding delta functions with low Fourier sparsity -- would improve the current best PIR schemes.

We show nontrivial upper and lower bounds on the Fourier sparsity of delta functions.
Our proofs are elementary and clean. These results imply limitations on improvements to the Matching Vector PIR schemes simply by finding better $S$-decoding polynomials. In particular, there are no $S$-decoding polynomials which can make Matching Vector PIRs based on the known matching vector families achieve polylogarithmic communication for constantly many servers.

Many interesting questions remain open.
\end{abstract}

\section{Introduction}

In this paper we study some basic and natural questions about Fourier analysis of Boolean functions. Specifically, we study the Fourier transform of functions $f$ defined on the whole group $\Z_m^r$ given constraints on the values that $f$ takes on the Boolean hypercube $\{0,1\}^r$. These questions were originally motivated by a connection to Matching Vector based Private Information Retrieval (PIR) schemes, but we feel that they are intrinsically interesting in their own right. 

The relationship of the subset $\{0,1\}^r$ with the superset $\Z_m^r$ is of central importance in many aspects of complexity theory. Many fundamental results about $\{0,1\}^r$ begin by viewing it as a subset of $\Z_m^r$ (or $\F_p^r$), and using the more powerful algebra of the larger domain. This includes the Razborov-Smolensky lower bounds for $AC_0$, arithmetization as a tool for interactive and probabilistically checkable proofs, and many more.

Let $m, r$ be integers. Our main concern is the Fourier behaviour of the Boolean hypercube $\{0,1\}^r$ contained within the abelian group $G = \Z_m^r$. 
Specifically, we will be interested in estimating the minimum possible {\em Fourier sparsity} of a {\em delta function} $f: G \to \mathbb C$.
\begin{itemize}
\item {\bf Fourier Sparsity: } For $i \in [m^r]$, let $\psi_i: G \to \C$ be the Fourier characters of $G$. Every function $f$ from $G$ to $\C$ can be expressed uniquely as a linear combination of the $\psi_i$. The number of $\psi_i$ that get a nonzero coefficient is called the {\em Fourier sparsity} of $f$.
\item {\bf Delta function: } A function $f : G \to \C$ is called a {\em delta function} if $f(0) = 1$, and for all $x \in \{0,1\}^r \setminus \{0\}$, we have $f(x) = 0$ (informally, it behaves like ``{\bf the} delta function"
on $\{0,1\}^r$).
\end{itemize}

On first reading, it may be useful to think about the asymptotic where $m$ is a fixed constant, even $3$, and $r \to \infty$.

The set of all delta functions is a $(m^r-2^r)$-dimensional affine subspace of the $m^r$-dimensional space of all functions, and thus by general linear-algebra principles, there is always a subset of $2^r$ Fourier characters whose span contains some delta function. Can this $2^r$ be reduced?

\subsection{Overview of Results}
When $m = 2$, there is only one delta function; and its Fourier sparsity is $2^r$.

Already for $m = 3$, the situation turns out to be quite interesting, with both nontrivial upper and lower bounds. We show that there are delta functions with Fourier sparsity $\leq O((\sqrt{3})^r)$, while any delta function must have Fourier sparsity at least $\Omega( (1.5)^r )$.

The situation drastically changes with growing $m$. We show that once $m \geq r$, 
there are delta functions with Fourier sparsity $\leq r+1$, and that this is tight!

For general $m$ (with $r \to \infty$), we show an upper bound of $O(m^{r/(m-1)}) = \exp\left( \frac{\log m}{m-1} \cdot r\right)$ and a lower bound of $\Omega\left(\left(\frac{m}{m-1}\right)^r\right) = \exp\left(\frac{1}{m-1}\cdot r\right)$.

\subsubsection*{More general groups}
For applications to PIRs (discussed below), we actually need to study  a slightly generalized problem.
\begin{itemize}
    \item The abelian group $G$ is now taken to be of the more general form $\prod_{i=1}^{r} \Z_{m_i}$, where $m_1, \ldots, m_r$ are positive integers,
    \item Instead of $\mathbb C$-valued Fourier characters, for a field $\mathbb F$ containing $m_i$ distinct $m_i$'th roots of unity, we work with $\F$-valued Fourier characters.
\end{itemize}
Even in this generalized setting, we can talk about the Boolean hypercube $\{0,1\}^r \subseteq G$ (which is simply the product of each $\{0,1\} \subseteq \Z_{m_i}$), delta functions, and the Fourier expansion of functions $f: G \to \F$.

Appropriate generalizations of our lower bounds from the case where all the $m_i$ are equal to $m$ continue to hold in this generalized setting. However our upper bounds break down, and for the setting relevant to PIRs, this leaves open a much bigger gap, which would be very interesting to close.

\subsubsection*{Delta functions on $\{-1,0,1\}^r$ }

One could also talk about delta functions on sets more general that $\{0,1\}^r$. 
For a subset $B \subseteq G$ with $0 \in B$, a $B$-delta function is a function $f$ with $f(0) = 1$ and $f(b) = 0$ for all $b \in B \setminus \{0\}$.

A natural and well known example is the case $B = G$. The unique $G$-delta function is the average of all its Fourier characters: this means that its Fourier sparsity is $|G|$.

A another natural set to consider is $B = \{-1, 0, 1\}^r$ (when $G = \prod_{i=1}^{r} \Z_{m_i}$). One would expect that the Fourier sparsity of $B$-delta functions to behave very similar to the case of $\{0,1\}^r$. Surprisingly, there is a big difference!

Again from general principles, since $|B| = 3^r$, there is always a $B$-delta function of sparsity at most $3^r$. Are there $B$-delta functions with smaller sparsity?

We show that for all $G$, the Fourier sparsity of a $B$-delta function is at least $2^r$. This is in sharp contrast to the $\{0,1\}^r$ case, where the sparsity decays sharply with $m$ in the group $G = \Z_m^r$, and there are are delta functions of sparsity as small as $r+1$ once $m > r$.

\subsubsection*{Techniques}

Our proofs are all elementary and simple. The lower bounds for $\{0,1\}^r$ are based on
studying the product of a purported Fourier sparse delta function with a suitably constructed auxiliary function with desirable Fourier and primal support. The upper bounds for $\{0,1\}^r$ are by explicitly giving the functions. The lower bounds for $\{-1, 0, 1\}^r$ are proved using the linear algebra method of combinatorics.

We feel that closing the gaps between our upper and lower bounds will need some new insights about the Fourier analysis of Boolean functions, which may potentially be useful in the many other situations in theoretical computer science where $\{0,1\}^r$ is viewed as a subset of the larger domain $\F_p^r$.

Another related theme  we would like to mention is the {\em uncertainty principle} in finite abelian groups, which gives lower bounds on the Fourier sparsity of functions $f: G \to \mathbb C$ with small support. See~\cite{DonohoStark1989,Meshulam2006,Tao2005} for various statements of this principle.

\subsection{Matching Vector PIRs}

The original motivation for studying Fourier-sparse delta functions is from questions about {\em Private Information Retrieval} (PIR). A PIR scheme\footnote{PIR schemes are very closely related to Locally Decodable Codes (LDCs). Many breakthroughs happened for both these notions together, but we only talk about PIRs in this paper. In particular, there are implications of our results for the LDCs known as Matching Vector codes.} is an interactive protocol between a user and $t$ non-communicating servers, which allows the user to access any desired bit $a_i$ of a database $(a_1, \ldots, a_n) \in \{0,1\}^n$ without revealing any information about $i$. PIRs have been extensively studied ever since their introduction by Chor, Goldreich, Kushilevitz and Sudan~\cite{CGKS} in 1995. 
The main question is to achieve this with as little communication as possible.
Very little is known on the lower bounds front; it is consistent with our knowledge that PIR can be achieved with $2$-servers and $\poly\log(n)$ communication complexity!

When the number of servers $t$ is constant, it was originally believed that the amount of communication needed to be at least polynomial in $n$, of the form $\Omega(n^{\varepsilon_t})$. However, many years later, surprising results of Yekhanin~\cite{Yekhanin} and Efremenko~\cite{Efremenko} strongly refuted this belief.
These protocols are the {\em Matching Vector PIRs}, and achieve a communication complexity of $O(\exp((\log n)^{\varepsilon_t}))$ -- and these are the only known way to achieve communication complexity subpolynomial in $n$ with a constant number of servers.

There are two key ingredients of Matching Vector PIRs:
\begin{itemize}
    \item {\bf An $S$-matching vector family:} for a subset $S \subseteq \Z_m$ with $0 \in S$, an $S$-matching vector family is a collection of $2n$ vectors $u_i, v_i \in \Z_m^k$, for $i \in [n]$, such that:
    \begin{itemize}
        \item $\langle u_i, v_j \rangle = 0$ if $i = j$,
        \item  $\langle u_i, v_j \rangle \in S \setminus \{0\}$ if $i \neq j$.
    \end{itemize}
    We would like $n$ as large as possible.
    \item {\bf An $S$-decoding polynomial:} for a subset $S \subseteq \Z_m$ and a field $\F$ containing an $m$-th root of unity $\gamma_m$, an $S$-decoding polynomial is a polynomial $P(Z) \in \F[Z]$ such that:
    \begin{itemize}
        \item  $P(\gamma^0) = 1$,
        \item  $P(\gamma^s) = 0$ for each $s \in S \setminus \{0\}$.
    \end{itemize}
    We would like $P(Z)$ to be as {\em sparse} as possible.
\end{itemize}
The best known PIR protocols~\cite{Efremenko,DGY,DG,GKS} use the $S$-matching vector families coming from the work of Barrington-Biegel-Rudich~\cite{BBR} and Grolmusz~\cite{Gro}. Here $m$ is taken to be the product of primes $m_1, \ldots, m_r$, and $S$ is taken to be a product set $\prod_{i=1}^r \{0,1\} \subseteq \prod_{i=1}^r \Z_{m_i} \simeq \Z_m$. Such $S$ is called a {\em canonical} $S$ for $\Z_m$.

Even with trivial $S$-decoding polynomials, the above $S$-matching vector family alone is enough to achieve the dramatic improvement of communication complexity from polynomial to subpolynomial. It achieves communication $\exp((\log n)^{\varepsilon_t})$. A clever choice of $S$-decoding polynomial further improves the communication complexity. The theory of sparse $S$-decoding polynomials was developed and advanced by Yekhanin~\cite{Yekhanin}, Raghavendra~\cite{Raghavendra}, Efremenko~\cite{Efremenko}, Itoh-Suzuki~\cite{IS} and Chee-Feng-Ling-Wang-Zhang~\cite{Cheeetal}, to achieve significant improvements for small $t$. Under plausible number theoretic conjectures,~\cite{Cheeetal}
reduces $\varepsilon_t$ by an absolute constant factor for all $t$.\footnote{The improvement is more stark for fixed $t$: the communication complexity using a clever $S$-decoding polynomial is a further subpolynomial in the communication complexity achieved using trivial $S$-decoding polynomials.}

To improve Matching Vector PIR parameters further, there are two immediate avenues.
The first is to get even larger $S$-matching vector families -- this is an extremely interesting and basic question about linear algebra mod $m$, and has attracted quite a bit of interest~\cite{BDL-matching,ADLOS25}. To summarize the current status: it is wide open whether these exist.

The other approach to improved PIRs, is to get sparser $S$-decoding polynomials for the known $S$-matching vector families. This approach, a priori, could also be capable of achieving polylogarithmic communication for a constant number of servers: this would be possible if there exist $S$-decoding polynomials of constant sparsity modulo any $m$. This is the approach that this paper addresses. 

Our results imply that for the $S$ relevant for the known big matching vector family over $\Z_m$, 
any $S$-decoding polynomial has sparsity at least $r+1$.
This means that for a Matching Vector PIR using the known matching vector families and the best possible $S$-decoding polynomial, the communication complexity for $t$ server PIR is at least $\exp((\log n)^{1/t})$. In particular it cannot achieve polylogarithmic communication for a constant number of servers.

\subsection{Questions}

There are two very nice and basic open questions that we would like to highlight.
\begin{enumerate}
\item Let $m_1, \ldots, m_r$ be distinct primes, and let $\F$ be a field containing all the $m_i$'th roots of $1$. We work with the $\F$-valued Fourier characters. How small, as a function of $r$, can the Fourier sparsity  of delta functions on $\{0,1\}^r \subseteq \prod_{i=1}^t\Z_{m_i}$?

We showed that it cannot be smaller than $r+1$, while the best upper bounds are exponential in $r$ (these come from known sparse $S$-decoding polynomials \cite{Cheeetal}).

Can it be $r+1$? This would improve PIR communication to $\exp((\log n)^{1/t})$.

Surprisingly, if we drop the assumption that the $m_i$ are distinct, 
then $r+1$ is achievable (see Lemma~\ref{ex:up-bnd}).

\item We have the same setup as above (with the $m_i$ being distinct primes), but now we fix $\F$ to be the complex numbers $\C$. We conjecture that every delta function on $\{0,1\}^r \subseteq \prod_{i=1}^r \Z_{m_i}$ has Fourier sparsity at least $2^r$: there is no improvement on the trivial bound!

We prove the $r =2$ case in Appendix~\ref{sec:appendix-mobius}. The proof uses ideas related to the classical Schwarz lemma from complex analysis, and boils down to studying certain M\"{o}bius transformations and their behavior on the unit disk. For larger $r$ it seems to be a challenging yet basic question.

Again, for identical $m_i$, Fourier sparsity $r+1$ is achievable, so the distinctness of the primes $m_i$ has to be used.

\end{enumerate}

Finally, the question of determining the Fourier sparsity of delta functions on $\Z_m^r$ for fixed $m$ is a very cleanly stated question which we do not fully understand. We know that it is exponential is $r$, but the base of the exponent is somewhere between $m^{\frac{1}{m-1}}$ and $\frac{m}{m-1}$.

\section{Main Result}
In order to state our main results, we need the following definitions.
\begin{definition}
Let $G$ be a group and $\F$ be a field, $0 \in B, B \subseteq G$. We say:
$$f: G \rightarrow \F$$ 
is a delta function on $B$ (or $f$ is delta on $B$, or $f$ is a $B$-delta function) if:
\begin{itemize}
\item $f(0) = 1$
\item For all $b \in B \setminus \{0\}$, we have $f(b) = 0.$

\end{itemize}
\end{definition}

For $\F$-valued functions $f$ on the abelian group $G$, we will be working with the $\F$-valued Fourier transform $\widehat{f}$ (generalizing the usual $\C$-valued Fourier transform), discussed in Section 3.

In the special case where the domain is $\Z_m^r$ and the function is delta on $\{0,1\}^r$, the following theorem provides explicit bounds on the Fourier sparsity.  
Parts (i) and (ii) are special cases of the general lower bounds proved later in Section~\ref{sec:bnd-boolean}, namely Theorems~\ref{thm-lower1} and \ref{thm-lower2}, while parts (iii) and (iv) correspond to the upper bounds established in Lemma~\ref{ex:up-bnd} and Claim~\ref{claim:upper-general}.  
The proofs of all four parts appear in Section~\ref{sec:bnd-boolean}.

\begin{theorem}
\label{thm:bounds-same-domain}
 Let $f: \Z_{m}^r \rightarrow \F$ be delta on $\{0, 1\}^r \subset \Z_{m}^r$. We have:
    \begin{enumerate}[label=(\roman*)]
        \item  \(|\Supp(\widehat{f})| \geq r + 1\)
        \item  \(|\Supp(\widehat{f})| \geq (\frac{m}{m - 1})^r\)
    \end{enumerate}

In the other direction:
\begin{enumerate}[label=(\roman*)]
        \item[(iii)]  For $m > r$, there exists $f : \Z_m^r \to \F$ which is delta on $\{0,1\}^r$ with
        $|\Supp(\widehat{f})| = r + 1$.
        \item[(iv)]  For $m, r$ with $(m-1) \mid r$, there exists $f: \Z_m^r \rightarrow \F$
        which is delta on $\{0, 1\}^r$ with $|\Supp(\widehat{f})| \leq m^\frac{r}{m - 1}$.
    \end{enumerate}
\end{theorem}

One application of our results is to $S$-decoding polynomials, a tool in the construction of Matching Vector PIRs. We describe this now. 

\begin{definition}[$S$-decoding polynomial for the canonical set $S$]
Let $m = p_1 \cdots p_r$ and define the canonical set
\begin{align*}
S &= \{x \in \Z_m : x^2  = x \}\\
&= \{ x \in \Z_m \mid  (x \mod p_i) \in \{0,1\}  \mbox{ for each } i\} .
\end{align*}
Let $\F$ be a field containing a primitive $m$-th root of unity $\gamma_m$.
A polynomial $P(Z) \in \F[Z]$ is called a $S$-\emph{decoding polynomial} if
\begin{itemize}
    \item $P(\gamma_m^s) = 0$ for all $s \in S \setminus \{0\}$,
    \item $P(\gamma_m^0) = P(1) = 1$.
\end{itemize}
If $P$ has exactly $t$ monomials, we say it is \emph{$t$-sparse}.
\end{definition}

We restrict here to the canonical set $S$; a more general definition for arbitrary subsets $S \subseteq \Z_m$ will be given in Section~\ref{poly-to-Fourier}.

$S$-decoding polynomials play a central role in the construction of Matching Vector PIR schemes: sparser polynomials yield the same communication complexity with fewer servers.  
For the canonical set $S \subseteq \Z_m$, our results show that such polynomials cannot be very sparse.  

The following theorem gives the precise lower bound; its proof appears in Section~\ref{poly-to-Fourier} as Corollary \ref{cor:poly-lower}.

\begin{theorem}
\label{thm:s-decoding}
Let $m = p_1 \cdots p_r$, and let $S$ be the canonical set for $m$.  
If $P(Z) \in \F[Z]$ is an $S$-decoding polynomial, then the number of its monomials is at least $r+1$.
\end{theorem}

The consequences for Matching Vector PIRs are discussed in Section~\ref{sec:pir}.

Our final main result concerns delta functions on a set different from the Boolean hypercube, namely $\{-1,0,1\}^r$.  
For such functions we obtain a much stronger lower bound on Fourier sparsity, that does not decay with $m$.  
The proof will be given in Section~\ref{sec:bnd-trinary} as a corollary of Claim~\ref{claim:bnd-diff}.  

\begin{theorem}
If
$f: \Z_{m}^r \rightarrow \F$ is delta on $\{-1, 0, 1\}^r$, then we have:
$$|\Supp(\widehat{f})| \geq 2^r$$  
\end{theorem}

\section{Bounds for Delta Functions}
In this section, we first derive bounds for delta functions on the Boolean hypercube $\{0,1\}^r$, and then we establish a lower bound for delta functions over $\{-1,0,1\}^r$.

\subsection{Basic Facts from Fourier Analysis}
We collect the Fourier-analytic facts used in this paper. 
Throughout, $G$ is a finite abelian group of size $m$, written additively.

For a function $h$ on a finite set $X$, we define its support as
\[
\Supp(h) \;:=\; \{\,x \in X : h(x) \neq 0\,\}.
\]
For subsets $A,B$ of an additive group, we write
\[
A+B \;:=\; \{\,a+b : a \in A,\ b \in B\,\}
\]
for their sumset.

We now define the Fourier transform and state the basic Fourier inversion formula.

\begin{definition}[Fourier transform over a finite abelian group, $\F$-valued]
Let $G$ be a finite abelian group, and let $e$ denote the exponent of $G$
(i.e., the least common multiple of the orders of elements of $G$).
Let $\F$ be a field with $\operatorname{char}(\F)\nmid |G|$ such that
$\F$ contains all $e$-th roots of unity (equivalently, a primitive $e$-th
root of unity, and hence all its powers).

A \emph{character} of $G$ is a group homomorphisms from $G$ to $\F^{\times}$, where $\F^\times$ is the multiplicative group of $\F$.
Let $\widehat{G}$ be the group of all characters of $G$.


For a function $f\colon G\to \F$, its Fourier transform is the function
$\widehat f\colon \widehat G\to \F$ defined by
\[
\widehat f(\psi)
\;=\; \frac{1}{|G|}\sum_{x\in G} f(x)\,\psi(x)^{-1},
\qquad \psi\in\widehat G.
\]
\end{definition}

Whenever we talk about the Fourier transform of $\F$-valued functions on $G$, we assume that $\operatorname{char}(|\F|) \not\mid |G|$, and that $\F$ has the required roots of unity in it. 
The first assumption is non-negotiable. The second assumption is without loss of generality, by replacing $\F$ with a finite extension containing the required roots of unity (this is possible because of the first assumption).

Next we state the basic fact about the Fourier transform of the product of functions.

\begin{lemma}[Character basis and explicit parametrization]\label{lem:fourier-basis}
Let $\F$ be a field and $G \cong \oplus_{i = 1}^{r} \Z_{m_i}$ be a finite abelian group such that 
$char(\F) \nmid |G|$ and 
$\F$ contains primitive $m_i$-th roots of unity $\omega_{m_i}$ for $1 \leq i \leq r$.

For any $a = (a_1,..., a_r) \in \oplus_{i = 1}^{r} \Z_{m_i}$, 
we have a character $\psi_a$ given by:
$$\psi_a(x_1,..., x_r) = \Pi_{i = 1}^r \omega_{m_i}^{a_i x_i}.$$
every character $\psi: G \rightarrow \F$ is of the form $\psi_a$ for some $a$.

Thus the character group $\widehat{G}$ is isomorphic to $\oplus_{i=1}^r \Z_{m_i}$,
and we identify them.

With this notation, every function
$f: G \rightarrow \F$
can be written as 
$$f = \sum_{a \in \widehat{G}} \widehat{f}(a).\psi_a$$
\end{lemma}

\begin{lemma}
    \label{conv}
Let $G$ be a finite abelian group and $\F$ be a field.  
Suppose $f,g: G \to \F$ are functions and define $h = f \cdot g$ as their pointwise product.  
Then the Fourier transform of $h$ is given by the convolution of the Fourier transforms of $f$ and $g$:
\[
\widehat{h}(a) = \bigl(\ \widehat{f} * \widehat{g}\,\bigr)(a) = \sum_{b \in \widehat{G}} \widehat{f}(b) \widehat{g}(a - b)
\]
This implies:
\[
\Supp(\widehat{h}) \subseteq \Supp(\widehat{f}) + \Supp(\widehat{g})
\]
\end{lemma}

\begin{corollary}
\label{product-supp}
Let $f,g : G \to \F$ be functions on a finite abelian group $G$, and set $h = f \cdot g$. 
Then
\[
|\Supp(\widehat{h})| \;\leq\; |\Supp(\widehat{f})| \, |\Supp(\widehat{g})|.
\]
\end{corollary}

Finally, we recall the Fourier expansion of the (unique) $G$-delta function,
which we call {\em the total delta function} to avoid ambiguity.

\begin{lemma}[Fourier expansion of the total delta function]
\label{delta}
Let $G$ be a finite abelian group and let $\F$ be a field with $\operatorname{char}(\F)\nmid |G|$.
Let $\delta$ denote the total delta function on $G$, i.e.
\[
\delta(x) \;=\;
\begin{cases}
1 & \text{if } x=0,\\
0 & \text{otherwise}.
\end{cases}
\]
Then its Fourier transform is given by
\[
\widehat{\delta}(a) \;=\; \frac{1}{|G|}
\qquad \text{for all } a \in \widehat G.
\]
In particular, $\widehat{\delta}$ is nonzero on every character, and so $\delta$ has Fourier sparsity exactly $|G|$.

\end{lemma}

\subsection{Bounds for Delta Functions on the Boolean Hypercube}\label{sec:bnd-boolean}

In what follows, we first prove two lower bounds for functions  
\[
f : \Z_{m_1} \times \cdots \times \Z_{m_r} \;\to\; \F
\]
that are delta on $\{0,1\}^r$. We then turn to the more specific setting 
\[
f : \Z_m^r \;\to\; \F,
\]
and establish an upper bound. In particular, when $m>r$, this upper bound matches the lower bound $r+1$, showing that the bound is tight.

\paragraph{First lower bound.}
Let 
\[
f : \Z_{m_1} \times \cdots \times \Z_{m_r} \;\to\; \F
\]
be a delta function on $\{0,1\}^r$. We present two proofs showing that its Fourier sparsity is at least $r+1$.

The following theorem formalizes this bound.

\begin{theorem}\label{thm-lower1}
Let 
\[
f : \Z_{m_1} \times \cdots \times \Z_{m_r} \;\to\; \F
\]
be a delta function on $\{0,1\}^r \subseteq \Z_{m_1} \times \cdots \times \Z_{m_r}$,
where $m_1,\dots,m_r$ are (not necessarily distinct) positive integers.  
Assume further that $\F$ contains $m_i$-th roots of unity $\omega_{m_i}$ for $1\leq i \leq r$.  
Then
\[
|\Supp(\widehat{f})| \;\geq\; r+1.
\]
\end{theorem}

\paragraph{First proof of Theorem~\ref{thm-lower1}.}
Our first proof proceeds by constructing an auxiliary function $g$ supported on $\{0,1\}^r$. 
Multiplying $f$ by $g$ will produce a total delta function, and from this we can deduce 
a structural relation between the Fourier supports of $f$ and $g$.
The following lemma makes this precise.

\begin{lemma}\label{lem:supp-sum}
Let 
\[
f : \Z_{m_1} \times \cdots \times \Z_{m_r} \;\to\; \F
\]
be a delta function on $\{0,1\}^r$, where $m_1,\dots,m_r$ are (not necessarily distinct) positive integers.  
Suppose $g : \Z_{m_1} \times \cdots \times \Z_{m_r} \to \F$ is a function such that
\begin{itemize}
    \item $g(x) = 0$ for all $x \notin \{0,1\}^r$, 
    \item $g(0) \neq 0$.
\end{itemize}
Then
\[
\Supp(\widehat{f}) + \Supp(\widehat{g}) \;=\; \Z_{m_1} \times \cdots \times \Z_{m_r}.
\]
\end{lemma}

\begin{proof}
Let $h = f g$. By Lemma~\ref{conv}, 
\[
\Supp(\widehat{h}) \;\subseteq\; \Supp(\widehat{f}) + \Supp(\widehat{g}).
\]
Since $g$ is supported on $\{0,1\}^r$ and $f$ vanishes on $\{0,1\}^r \setminus \{0\}$ with $f(0)\neq 0$, we see that $h$ is supported only at $0$. Thus $h$ is a nonzero scalar multiple of the total delta function on the whole group, and so by Lemma \ref{delta},
\[
\Supp(\widehat{h}) = \Z_{m_1} \times \cdots \times \Z_{m_r}.
\]
This gives
\[
\Supp(\widehat{f}) + \Supp(\widehat{g}) \;=\; \Z_{m_1} \times \cdots \times \Z_{m_r}.
\]
\end{proof}

To apply Lemma~\ref{lem:supp-sum}, we need an explicit function $g$ satisfying 
the required conditions and whose Fourier support we can describe precisely. 
The following example provides such a construction.

\begin{example}\label{ex:comp-g} 
Let $\F$ be a field containing primitive $m_i$-th roots of unity $\omega_{m_i}$ for $1 \leq i \leq r$. Define 
\[ g \colon \Z_{m_1}\times \cdots \times \Z_{m_r} \;\longrightarrow\; \F \]
by 
\[ g(x_1,\dots,x_r) = \prod_{i=1}^r \; h_i(x_i),\] 
where, for each $i \in [r]$:
\[ h_i(x_i) \;=\; \omega_{m_i}^{x_i}\;\prod_{j=2}^{m_i-1} \bigl(\omega_{m_i}^{x_i}-\omega_{m_i}^j \bigr).\]

First we see that the values of $g$ are as needed for Lemma~\ref{lem:supp-sum}. 

Observe that:
\begin{itemize}
\item $h_i(0) \neq 0$,
\item $h_i(2) = h_i(3) = \ldots = h_i(m_i-1) = 0$.
\end{itemize}

Thus:
\begin{itemize}
    \item $g(0) \neq 0$,
    \item $g(x) = 0$ for all $x \not\in \{0,1\}^r$.
\end{itemize}

Now we understand the Fourier spectrum of $g$.

Observe that $h_i$ is a linear combination of the $m_i-1$ characters
\(\{\;\omega_{m_i}^{b_i x_i} : 1 \leq b_i \leq m_i-1 \;\}\).

This means that $g$ lies in the span of the functions $\Bigl\{\, \prod_{i=1}^r \omega_{m_i}^{b_i x_i} \;\;:\;\; 1 \leq b_i \leq m_i-1 \Bigr\} = \{ \psi_b \mid b  = (b_1, \ldots, b_r),  1 \leq b_i \leq m_i-1\}$. 
It follows that the Fourier support of $g$ is exactly
\[ \Supp(\widehat g) = \prod_{i=1}^r \{1,\dots,m_i-1\}, \] 
and therefore
\[ |\Supp(\widehat g)| \;=\; \prod_{i=1}^r (m_i-1). \] \end{example}

We are now ready to give the first proof of Theorem~\ref{thm-lower1}.\begin{proof}[Proof of Theorem~\ref{thm-lower1}]
By Lemma~\ref{lem:supp-sum} and Example~\ref{ex:comp-g}, we know that
\[
\Supp(\widehat{f}) + \prod_{i=1}^r \{1,\dots,m_i-1\} 
= \prod_{i=1}^r \{0,1,\dots,m_i-1\}.
\]

This implies the following property: for every element 
\[
u \in \prod_{i=1}^r \{0,1,\dots,m_i-1\},
\]
there must exist some $a \in \Supp(\widehat{f})$ such that 
$u_i \neq a_i$ for all coordinates $i$. 

Now assume, for the sake of contradiction, that 
\[
\Supp(\widehat{f}) = \{b^1,\dots,b^t\}, \qquad t \leq r.
\]
That is, suppose the Fourier support has size at most $r$.  

We will use a form of ``diagonalization" to derive a contradiction.
Construct a vector $u \in \prod_{i=1}^r \{0,1,\dots,m_i-1\}$ by setting
\[
u = (b^1_1, b^2_2, \dots, b^t_t, 0,\dots,0).
\]
In words: for $i \leq t$, the $i$-th coordinate of $u$ equals the $i$-th coordinate of $b^i$ (and the remaining coordinates of $u$ are set to $0$).  

By the above mentioned property, there must exist some $b \in \Supp(\widehat{f})$ 
such that $b_j \neq u_j$ for every coordinate $j$.  
But by construction, for each $i \leq t$ we have $u_i = b^i_i$, which forces 
$b \neq b^i$ for all $i=1,\dots,t$.  
This contradicts the assumption that $\Supp(\widehat{f}) = \{b^1,\dots,b^t\}$.  

Therefore, our assumption was false, and we conclude that
\[
|\Supp(\widehat{f})| \;\geq\; r+1.
\]
\end{proof}

\paragraph{Second proof of Theorem~\ref{thm-lower1}.}
Unlike the first proof, which relied on support-sum arguments, 
our second proof of Theorem~\ref{thm-lower1} uses a polynomial identity 
arising from the delta property of $f$ on $\{0,1\}^r$.
\begin{proof}
By assumption, there exists $\omega_{m_i} \in \F$ where each $\omega_{m_i}$ is an $m_i$-th root of unity.  
By Lemma~\ref{lem:fourier-basis}, we can expand $f$ as
\[
f(x_1,\dots,x_r) \;=\; \sum_{j=1}^t b_j \prod_{i=1}^r \omega_{m_i}^{\,a_j(i)\,x_i}.
\]

For convenience, define the vectors $\alpha_i \in \F^t$ by
\[
\alpha_i = \bigl(\alpha_i(1),\dots,\alpha_i(t)\bigr), \qquad
\alpha_i(j) := \omega_{m_i}^{\,a_j(i)} \quad \text{for } 1 \leq j \leq t.
\]
With this notation we can rewrite $f$ as
\[
f(x_1,\dots,x_r) \;=\; \sum_{j=1}^t b_j \prod_{i=1}^r \alpha_i(j)^{\,x_i}.
\]

Since $f$ is a delta function on $\{0,1\}^r$, we have
\[
f(x)=0 \quad \text{for all } x \in \{0,1\}^r \setminus \{0\}, 
\qquad f(0,\dots,0) \neq 0.
\]
Equivalently,
\[
\sum_{j=1}^t b_j \prod_{i \in B} \alpha_i(j) \;=\; 0
\quad \text{for every nonempty } B \subseteq [r],
\qquad \sum_{j=1}^t b_j \;\neq\; 0.
\]

Observe that, introducing formal variables $T_1,\dots,T_r$, expanding
\[
\prod_{i=1}^r \bigl(1 + T_i \alpha_i(j)\bigr)
\]
produces terms of the form $\prod_{i \in B} T_i \alpha_i(j)$ for subsets $B \subseteq [r]$. 
Therefore
\[
\sum_{j=1}^t b_j \prod_{i=1}^r (1 + T_i \alpha_i(j))
= \sum_{B \subseteq [r]} \Bigl( \prod_{i \in B} T_i \Bigr) 
   \Bigl( \sum_{j=1}^t b_j \prod_{i \in B} \alpha_i(j) \Bigr).
\]
By the vanishing property established above, every inner sum with $B \neq \emptyset$ equals zero, while for $B=\emptyset$ we get $\sum_{j=1}^t b_j \neq 0$. 
This proves the polynomial identity
\[
\sum_{j=1}^t b_j \prod_{i=1}^r (1 + T_i \alpha_i(j))
= \sum_{j=1}^t b_j.
\]
If $r \geq t$, we may choose $T_i = -1/\alpha_i(i)$ for $1 \leq i \leq t$.  
With this choice, each term in the left-hand side vanishes, so the entire left-hand side equals $0$, while the right-hand side is nonzero. This yields a contradiction.
\end{proof}

\paragraph{Second lower bound.}
We again use Lemma~\ref{lem:supp-sum} and Example~\ref{ex:comp-g}, but this time in combination with a counting argument, to derive a multiplicative lower bound.

\begin{theorem}\label{thm-lower2}
Let 
\[
f : \Z_{m_1} \times \cdots \times \Z_{m_r} \;\to\; \F
\]
be a delta function on $\{0,1\}^r \subseteq \Z_{m_1} \times \cdots \times \Z_{m_r}$,
where $m_1,\dots,m_r$ are (not necessarily distinct) positive integers.  
Then
\[
|\Supp(\widehat{f})| \;\geq\; \prod_{i=1}^r \frac{m_i}{\,m_i-1\,}.
\]
\end{theorem}

\begin{proof}
By Lemma~\ref{lem:supp-sum} and Example~\ref{ex:comp-g}, we have
\[
\Supp(\widehat{f}) + B = \prod_{i=1}^r \{0,1,\dots,m_i-1\},
\qquad 
B := \prod_{i=1}^r \{1,\dots,m_i-1\}.
\]
  
Since $A+B$ has size at most $|A| \cdot |B|$ for any sets $A,B$, we obtain
\[
|\Supp(\widehat{f})| \;\geq\; \frac{|\Supp(\widehat{f})+B|}{|B|}.
\]
Here $|\Supp(\widehat{f})+B| = \prod_{i=1}^r m_i$, while 
$|B| = \prod_{i=1}^r (m_i-1)$.  
Therefore
\[
|\Supp(\widehat{f})| \;\geq\; \frac{\prod_{i=1}^r m_i}{\prod_{i=1}^r (m_i-1)}
= \prod_{i=1}^r \frac{m_i}{m_i-1}.
\]
\end{proof}

\begin{corollary}
\label{same-lower}
If 
$f: \Z_m^r \rightarrow \F$ is delta on $\{0, 1\}^r$, then:
\begin{enumerate}[label=(\roman*)]
    \item  \(|\Supp(\widehat{f})| \geq r + 1\)
    \item  \(|\Supp(\widehat{f})| \geq (\frac{m}{m - 1})^r\)
\end{enumerate}
\end{corollary}

\paragraph{Upper bound.}
We next study upper bounds for the Fourier sparsity of delta functions on $\{0,1\}^r$. The following construction illustrates how to achieve small support when $m>r$, and will serve as a building block for further results in this section.

\begin{lemma}\label{ex:up-bnd}
Let $m > r$ and let $\omega \in \F$ be an $m$-th root of unity. 
Then there exists a function
\[
f : \Z_m^r \;\to\; \F
\]
which is a delta function on $\{0,1\}^r \subseteq \Z_m^r$ and has exactly $r+1$ nonzero Fourier coefficients.
\end{lemma}
    
\begin{proof}
Define
\[
f(x_1,\dots,x_r) \;=\; \prod_{i=1}^{r} \bigl(\,\omega^{\,x_1+\cdots+x_r} - \omega^i \bigr).
\]
We first check the delta property.  
For any $x \in \{0,1\}^r \setminus \{0\}$ we have $1 \leq x_1+\cdots+x_r \leq r$, hence one of the factors vanishes, and thus $f(x)=0$.  
On the other hand, $f(0,\dots,0) = \prod_{i=1}^r (1 - \omega^i) \neq 0$.  
Therefore $f$ is indeed a delta function on $\{0,1\}^r$.

Finally, observe that $f$ can be expressed as a linear combination of the $r+1$ characters
\[
\psi_{(i,\dots,i)}(x_1,\dots,x_r) \;=\; \omega^{\,i(x_1+\cdots+x_r)}, 
\qquad 0 \leq i \leq r.
\]
Thus $f$ has exactly $r+1$ nonzero Fourier coefficients, as claimed.
\end{proof}

\begin{corollary}
When $m > r$, the above claim gives a construction of a delta function with Fourier sparsity $r+1$. 
By Corollary~\ref{same-lower}, we also have the lower bound $r+1$. 
Thus, in this case, the upper and lower bounds coincide, yielding a tight result.
\end{corollary}

\begin{theorem}\label{claim:upper-general}
Suppose $m - 1$ divides $r$. Then, there exists a function
\[
f : \Z_m^r \;\to\; \F
\]
which is a delta function on $\{0,1\}^r$ and satisfies
\[
|\Supp(\widehat{f})| \;\leq\; m^{\,\tfrac{r}{m-1}}.
\]
\end{theorem}

\begin{proof}
Partition the index set $\{1,\dots,r\}$ into 
\[
\ell \;=\; \frac{r}{m-1}
\]
disjoint subsets, denoted $A_1,\dots,A_\ell$, each of size $m-1$.  
For each $1 \leq k \leq \ell$, define a function $f_k$ depending only on the coordinates indexed by $A_k$, as in Lemma~\ref{ex:up-bnd}:
\[
f_k\bigl((x_j)_{j \in A_k}\bigr) 
\;=\; \prod_{i=1}^{m-1}\Bigl(\,\omega^{\sum_{j \in A_k} x_j} - \omega^i \Bigr),
\]
where $\omega$ is a fixed primitive $m$-th root of unity.  
By Lemma~\ref{ex:up-bnd}, each $f_k$ is delta on 
$\{0,1\}^{A_k}$
and satisfies
\[
|\Supp(\widehat{f_k})| \;=\; m.
\]

Now define the overall function
\[
f(x_1,\dots,x_r) \;=\; \prod_{k=1}^{\ell} f_k\bigl((x_j)_{j \in A_k}\bigr).
\]
Clearly, $f$ is delta on $\{0,1\}^r$.  
Applying Corollary ~\ref{product-supp} , we obtain
\[
|\Supp(\widehat{f})| 
\;\leq\; \prod_{k=1}^\ell |\Supp(\widehat{f_k})|
\;=\; m^\ell 
\;=\; m^{\,\tfrac{r}{m-1}}.
\]
\end{proof}
\subsection{Bounds for Delta Functions on $\{-1,0,1\}^r$}\label{sec:bnd-trinary}
We now turn to delta functions on $\{-1,0,1\}^r$. Using a general structural claim, we derive lower bounds on their Fourier sparsity. In particular, we show that such functions must have support size at least $2^r$.

\begin{claim}\label{claim:bnd-diff}
Let $f : G \to \F$ be a delta function on a set $B \subseteq G$ with $0 \in B$.  
Suppose $D \subseteq B$ is such
$$D-D \subseteq B,$$
where $D-D$ is the difference set $\{ d_1 - d_2 \mid d_1, d_2 \in D \}$.
Then
\[
|\Supp(\widehat{f})| \;\geq\; |D|.
\]
\end{claim}

\begin{proof}
Write $f$ as a linear combination of characters supported on some set 
$T \subseteq \widehat{G}$:
\[
f(x) \;=\; \sum_{a \in T} c_a \psi_a(x).
\]

For $d \in G$, define the translate $f_d: G \to \F$ by:
\[
f_d(x) := f(x- d).
\]

Then we compute the Fourier expansion of $f_d$ as follows:
$$f_d(x) = \sum_{a \in T} c_a \psi_a(x-d) = \sum_{a \in T} \left(c_a \psi^{-1}_a(d) \right)\psi_a(x),$$
and thus:
\[
f_d \in \Span\{\psi_a : a \in T\}.
\]

We now show that the functions $\{f_d : d \in D \}$ are all linearly independent.
Enumerate $D = \{d_1,\dots,d_{|D|}\}$.  
For each $d \in D$, consider the evaluation vector
\[
v_d := \bigl(f_d(d_1),\dots,f_d(d_{|D|})\bigr) \;\in\; \F^{|D|}.
\]

By construction,
\[
f_d(d_i) = f(d-d_i).
\]
If $d \neq d_i$, then $d-d_i \in B \setminus \{0\}$, so $f(d-d_i)=0$.  
If $d = d_i$, then $f_d(d_i)  = f(0) \neq 0$.  
Thus $v_d$ is a nonzero scalar multiple of the $i$-th standard basis vector $e_i$.
Hence the vectors $\{v_d : d \in D\}$ are linearly independent.

Therefore, the functions $\{f_d : d \in D\}$ are linearly independent, all 
lying inside $\Span\{\psi_a : a \in T\}$.  
It follows that
\[
|T| \;\geq\; |D|.
\]

Finally, since $T = \Supp(\widehat{f})$, we conclude that 
$|\Supp(\widehat{f})| \geq |D|$, as claimed.
\end{proof}


\begin{corollary}
If
$f: \Z_{m_1} \times ... \times \Z_{m_r} \rightarrow \F$ is delta on $\{-1, 0, 1\}^r$, then we have:
$$|\Supp(\widehat{f})| \geq 2^r$$    
\end{corollary}

\begin{proof}
Apply the previous claim with 
$B=\{-1,0,1\}^r$ and 
$D=\{0, 1\}^r$.  
For distinct $d_1,d_2 \in D$, the difference $d_1-d_2$ lies in $B \setminus \{0\}$, 
so the conditions of the claim are satisfied.  
Since $|D|=2^r$, the result follows.
\end{proof}

\section{From Fourier Sparsity to Matching Vector PIRs}
\label{sec:pir}

In this section, we first establish the connections between the Fourier sparsity of delta functions and $S$-decoding polynomials. We then recall how the sparsity of $S$-decoding polynomials is related to Matching vector PIRs. Combining this with our results from the previous section we derive new bounds on the number of servers required in such PIR schemes.

\subsection{Fourier Sparsity and $S$-decoding Polynomials}\label{poly-to-Fourier}

For Matching Vector PIRs, there are two key ingredients: a large $S$-matching vector family in $\Z_m^k$, and a sparse $S$-decoding polynomial, where $S \subseteq \Z_m$.

For all the superpolynomially large $S$-matching vector families that we know\footnote{They all come from the Barrington-Biegel-Rudich method involving the Chinese remainder theorem.}, $S$ has to have a special form: it must essentially be a canonical set for $m$, defined below\footnote{More precisely, $S$ must contain a product set of the form $S_1 \times S_2 \ldots \times S_r$, where  $\S_i \subseteq \Z_{p_i}$ has size at least $2$, and $\Z_m$ is viewed as $\prod_i \Z_{p_i}$ for distinct primes $p_1, \ldots, p_r$. $S$-decoding polynomials for such $S$ easily translate into $S$-decoding polynomials for the canonical $S$.}.

\begin{definition}[Canonical set]
\label{def:canonicalset}
Let 
$m \in \Z$
be a positive integer. We define the canonical set for 
$m$
to be the following set:
$$S = \{x \in \Z_m \ s.t. \ x^2 \equiv x \ mod \ m\}$$

Note that when $m = p_1 \cdot \cdot \cdot p_r$ is a product of $r$ distinct primes, the canonical set consists of the following $2^r$ values:
$$S = \{s \in \Z_m,  \forall p_i \ s \ mod \ p_i \in \{0, 1\} \}$$
For example, for $m = 6$ we get:
$$S = \{0, 1, 3, 4\}$$
which are precisely the
four values being
$0$ or $1$ 
mod $2, 3$.
\end{definition}

\begin{definition}[$t$-sparse $S$-decoding polynomial]
Let $S \subseteq \Z_m$ be a set containing $0$,
$\F$ be a field containing an element $\gamma_m$
which is a primitive $m$-th root of $1$: namely
$\gamma_m^m = 1$ 
and 
$\gamma_m^i \neq 1$
for 
$i = 1, 2,..., m - 1 $.

A polynomial 
$P(Z) \in \F[Z]$
is called an 
$t$-sparse $S$-decoding polynomial if $P$ has $t$ monomials, and the following conditions hold:
\begin{itemize}
    \item $\forall s \in S \setminus \{0\}: \ P(\gamma_m^s) = 0$
    \item $P(\gamma_m^0) = P(1) = 1$
\end{itemize}
\end{definition}

For $m = p_1 \times...\times p_r$, for distinct primes $p_1,..., p_r$, we use $\phi$ to denote the Chinese remainder isomorphism:
$$\phi: \Z_m \rightarrow \Z_{p_1}\times...\times\Z_{p_r}$$
$$\phi(a) = (a(1),..., a(r)), \ a(i) = a \ mod \ p_i$$
Also, we use
$\phi(a)_i$ to denote $a(i) = a \ mod \ p_i$.

The next claim shows the relation between Fourier sparsity and $S$-decoding polynomials for canonical sets $S$.

\begin{claim}
Let
$m = p_1...p_r$ be a product of $r$ distinct primes and $S$ be the canonical set for $m$,
$\F$ be a field containing an
$m$-th root of unity. Then, having a 
$t$-sparse 
$S$-decoding polynomial in $\F[Z]$ is equivalent to having a function:
$$f: \Z_{p_1} \times ... \times \Z_{p_r} \rightarrow \F$$
such that
$$f|_{\{0, 1\}^r}$$
is a delta function,
$|\Supp(\widehat{f})| = t.$
\end{claim}

\begin{proof}
Let
$P(Z) \in \F[Z]$
be an
$S$-decoding polynomial,
$$P(Z) = \sum_{j = 1}^{t} b_j Z^{a_j}$$
Let
$\gamma_m \in \F$
be a primitive 
$m$-th
root of unity. Then, since $p_i$'s are pairwise coprime, we can write it as:
$$\gamma_m = \prod_{i = 1}^r \omega_{p_i}$$
for $\omega_{p_i} \in \F$, $1 \leq i \leq r$, with each $\omega_{p_i}$ a primitive $p_i$-th root of unity. Therefore,
$$\gamma_m^s = (\prod_{i = 1}^r \omega_{p_i})^s = \prod_{i = 1}^r \omega_{p_i}^{\phi(s)_i}$$
\\
Observe that:
$$P(\gamma_m^s) = 0 \iff \sum_{j=1}^t b_j(\prod_{i = 1}^r \omega_{p_i}^{\phi(s)_i})^{a_j} = 0 \iff \sum_{j=1}^t b_j\prod_{i = 1}^r \omega_{p_i}^{\phi(a_j)_i \cdot \phi(s)_i} = 0$$
Now note that 
\[
S = \{s \in \Z_m: s \ mod \ p_i \in \{0, 1\}\} = \{s \in \Z_m: \phi(s) \in \{0, 1\}^r\}
\]
Hence, having such an $S$-decoding polynomial is equivalent to constructing a function $f$ such that:
$$f(x_1,..., x_r) = \sum_{j=1}^t b_j\prod_{i = 1}^r \omega_{p_i}^{a_j(i)x_i},$$
and:
$$f|_{(\phi(S)\setminus \{0\})} = 0,$$
$$f(0) \neq 0.$$
By Lemma \ref{lem:fourier-basis}, $ \prod_{i=1}^r \omega_{p_i}^{a_j(i) x_i} = \psi_{a_j}(x)$,
and so the above expression for $f$ equals:
$$ f(x) = \sum_{j=1}^t b_j \psi_{a_j}(x).$$

Thus having such an $S$-decoding polynomial is exactly equivalent to having a function 
$f: \Z_m \cong \Z_{p_1} \times...\times \Z_{p_r} \rightarrow \F$ such that
$f|_{\phi(S)}$ is delta and $|\Supp(\widehat{f})| = |\{b_j : j \in [t]\}| = t$.

\end{proof}

By applying Theorem \ref{thm-lower1} to the above claim, we obtain the following corollary:
\begin{corollary}
\label{cor:poly-lower}
Let $m = p_1 \times... \times p_r$, $S$ be the canonical set for $m$, $P(Z) \in \F[Z]$ be an $S$-decoding polynomial. Then, the number of monomials of $P$ is at least $r + 1$.       
\end{corollary}

\subsection{Applications to Matching Vector PIRs}

We now discuss the relationship between $S$-decoding polynomials and Matching Vector based PIRs.

\begin{definition}[$S$-Matching Vector Family]
Let
$S \subset \Z_m, 0 \in S$
and
$\Fam = \left( \U, \V \right)$
where 
$\U = \left(u_1,..., u_n \right), \V = \left(u_1,..., u_n \right)$ with  $u_i, v_i \in \Z_m^k$.
Then 
$\Fam$
is said to be an 
$S$-matching 
vector family of size 
$n$
and dimension
$k$
if the following conditions hold:
\begin{itemize}
    \item $\langle u_i, v_i \rangle = 0$ for every $i \in [n]$
    \item $\langle u_i, v_j \rangle \in S \setminus \{0\}$ for every $i \neq j$
\end{itemize}
\end{definition}

The connection between decoding polynomials and matching vector PIR schemes is made explicit by the following result:

\begin{theorem}[Theorem 2.3 in \cite{GKS}]
\label{MV-PIR}
    Let $m$ be a positive integer with $r$ distinct prime factors, and let $p$ be a prime not dividing $m$.  
Set $M = mp$, and let
\[
\mathcal{F} = (\mathcal{U}, \mathcal{V})
\]
be an $S_M$-matching vector family over $\Z_M$ of size $n$ and dimension $k$, where $S_M$ denotes the canonical set for $\Z_M$.

Suppose $\F$ is a field of characteristic $p$ such that, for the canonical set $S_m \subseteq \Z_m$, there exists an $S_m$-decoding polynomial over $\F$ with sparsity $t$.

Then there exists a $t$-server PIR scheme with communication complexity $O(k)$.

\end{theorem}

The above theorem highlights the role of $S$-decoding polynomials in such PIR schemes.  
In particular, designing sparser $S$-decoding polynomials leads directly to improved PIR constructions: 
The same communication complexity can be achieved with fewer servers.  
In other words, the sparsity of decoding polynomials dictates the number of servers required.

Theorem~\ref{thm:MVF} gives the best known construction of large matching vector families, where the set $S$ is taken to be the canonical set.  

\begin{theorem}[Large Matching Vector Families, Theorem 1.4 in \cite{Gro}]
\label{thm:MVF}
Let 
$m = p_1p_2...p_r$ where $p_1, p_2,..., p_r$ are distinct
primes and $r \geq 2$. Then, there exists an explicitly constructible $S$-matching vector family 
$\Fam$ 
in
$\Z_m^k$
of size 
$n \geq exp\left( \Omega\left(\frac{(log \ k)^r}{(log \ log \ k)^{r - 1}} \right)  \right)$
where 
$S$
is the canonical set for $m$.
\end{theorem}

Combining this with Theorem~\ref{MV-PIR}, we obtain that any $S$-decoding polynomial for the canonical set directly yields a PIR scheme with parameters determined by the sparsity of the polynomial.  
In particular, a $t$-sparse decoding polynomial for the canonical $S$ gives rise to a $t$-server PIR scheme with communication complexity $\exp(\Otilde( ( \log n)^{\frac{1}{r+1}}))$ and database size $n$.

For the canonical set $S \subseteq \Z_m$, Corollary~\ref{cor:poly-lower} shows that every $S$-decoding polynomial has sparsity at least $r+1$.  
Consequently, the Matching Vector PIR scheme based on the above matching vector family requires at least $r+1$ servers.  

Stating this another way, the communication complexity of a $t$ server Matching Vector PIR scheme based on the this matching vector family and an arbitrary $S$-decoding polynomial is bounded below by
\[
\exp\!\left(\bigl((\log n)^{1/t}\bigr)\right).
\]
Hence these PIR schemes cannot achieve polylogarithmic communication with a constant number of servers merely by improving the $S$-decoding polynomials.

It would be extremely interesting to find $S$-decoding polynomials matching this sparsity bound. Today, the best known sparsity comes from the results of~\cite{Cheeetal}, and is approximately is $(\sqrt{3})^r$ under a plausible number theoretic hypothesis.

\bibliographystyle{alpha}
\bibliography{pir}

@article{DonohoStark1989,
  author    = {Donoho, David L. and Stark, Philip B.},
  title     = {Uncertainty Principles and Signal Recovery},
  journal   = {SIAM Journal on Applied Mathematics},
  volume    = {49},
  number    = {3},
  pages     = {906--931},
  year      = {1989},
  doi       = {10.1137/0149053}
}

@article{Meshulam2006,
  author    = {Meshulam, Roy},
  title     = {An Uncertainty Inequality for Finite Abelian Groups},
  journal   = {European Journal of Combinatorics},
  volume    = {27},
  number    = {1},
  pages     = {63--67},
  year      = {2006},
  doi       = {10.1016/j.ejc.2005.01.004}
}

@article{Tao2005,
  author    = {Tao, Terence},
  title     = {An Uncertainty Principle for Cyclic Groups of Prime Order},
  journal   = {Mathematical Research Letters},
  volume    = {12},
  number    = {1},
  pages     = {121--127},
  year      = {2005},
  doi       = {10.4310/MRL.2005.v12.n1.a10}
}

@article{CGKS,
author = {Chor, Benny and Goldreich, Oded and Kushilevitz, Eyal and Sudan, Madhu},
title = {{Private information retrieval}},
year = {1998},
publisher = {Association for Computing Machinery},
volume = {45},
pages = {965–981},
}

@book{RudinRealComplexAnalysis,
  author    = {Rudin, Walter},
  title     = {Real and Complex Analysis},
  edition   = {3},
  publisher = {McGraw-Hill},
  address   = {New York},
  year      = {1987},
  isbn      = {978-0070542341}
}

@inproceedings{ADLOS25,
  author       = {Divesh Aggarwal and
                  Pranjal Dutta and
                  Zeyong Li and
                  Maciej Obremski and
                  Sidhant Saraogi},
  editor       = {Raghu Meka},
  title        = {Improved Lower Bounds for 3-Query Matching Vector Codes},
  booktitle    = {16th Innovations in Theoretical Computer Science Conference, {ITCS}
                  2025, January 7-10, 2025, Columbia University, New York, NY, {USA}},
  series       = {LIPIcs},
  volume       = {325},
  pages        = {2:1--2:19},
  publisher    = {Schloss Dagstuhl - Leibniz-Zentrum f{\"{u}}r Informatik},
  year         = {2025},
  url          = {https://doi.org/10.4230/LIPIcs.ITCS.2025.2},
  doi          = {10.4230/LIPICS.ITCS.2025.2},
  bibsource    = {dblp computer science bibliography, https://dblp.org}
}

@inproceedings{BDL-matching,
author = {Bhowmick, Abhishek and Dvir, Zeev and Lovett, Shachar},
title = {New bounds for matching vector families},
year = {2013},
isbn = {9781450320290},
publisher = {Association for Computing Machinery},
address = {New York, NY, USA},
url = {https://doi.org/10.1145/2488608.2488713},
doi = {10.1145/2488608.2488713},
booktitle = {Proceedings of the Forty-Fifth Annual ACM Symposium on Theory of Computing},
pages = {823–832},
numpages = {10},
location = {Palo Alto, California, USA},
series = {STOC '13}
}

@article{Yekhanin,
author = {Yekhanin, Sergey},
title = {{Towards 3-query locally decodable codes of subexponential length}},
year = {2008},
publisher = {Association for Computing Machinery},
volume = {55},
number = {1},
journal = {J. ACM},
articleno = {1},
numpages = {16},
}

@inproceedings{Efremenko,
author = {Efremenko, Klim},
title = {{3-query locally decodable codes of subexponential length}},
year = {2009},
publisher = {Association for Computing Machinery},
pages = {39–44},
numpages = {6},
series = {STOC '09}
}

@article{IS,
  title={{Improved Constructions for Query-Efficient Locally Decodable Codes of Subexponential Length}},
  author={Toshiya Itoh and Yasuhiro Suzuki},
  journal={IEICE Trans. Inf. Syst.},
  year={2008},
  volume={93-D},
  pages={263-270},
}

@article{Cheeetal,
    title = {{Query-Efficient Locally Decodable Codes of Subexponential Length}},
    author = {Chee, Yeow Meng and Feng, Tao and Ling, San and Wang, Huaxiong and  Feng Zhang, Liang},
    journal = {Computational Complexity},
    year = {2013},
    volume = {22},
    pages = {159–189},
}

@article{Raghavendra,
  title={{A Note on Yekhanin's Locally Decodable Codes}},
  author={Prasad Raghavendra},
  journal={Electron. Colloquium Comput. Complex.},
  year={2007},
  volume={TR07},
}

@inproceedings{DG,
author = {Dvir, Zeev and Gopi, Sivakanth},
title = {{2-Server PIR with Sub-Polynomial Communication}},
year = {2015},
publisher = {Association for Computing Machinery},
pages = {577–584},
series = {STOC '15}
}

@INPROCEEDINGS{DGY,
  author={Dvir, Zeev and Gopalan, Parikshit and Yekhanin, Sergey},
  publisher={FOCS}, 
  title={{Matching Vector Codes}}, 
  year={2010},
  pages={705-714},
}

@article{Gro,
    title = {{Superpolynomial size set-systems with restricted intersections
    mod 6 and explicit Ramsey graphs}},
    author = {Grolmusz, Vince},
    publisher = {Combinatorica},
    volume = {20},
    year = {2000}
}

@article{BBR,
  author       = {David A. Mix Barrington and
                  Richard Beigel and
                  Steven Rudich},
  title        = {Representing {B}oolean Functions as Polynomials Modulo Composite Numbers},
  journal      = {Comput. Complex.},
  volume       = {4},
  pages        = {367--382},
  year         = {1994},
  url          = {https://doi.org/10.1007/BF01263424},
  doi          = {10.1007/BF01263424},
  bibsource    = {dblp computer science bibliography, https://dblp.org}
}

@INPROCEEDINGS{GKS,
author = {Ghasemi, Fatemeh and Kopparty, Swastik and Sudan, Madhu},
title = {Improved PIR Schemes using Matching Vectors and Derivatives},
year = {2025},
isbn = {9798400715105},
url = {https://doi.org/10.1145/3717823.3718313},
doi = {10.1145/3717823.3718313},
pages = {1648–1656},
numpages = {9},
series = {STOC '25}
}

\appendix

\section{A more sensitive lower bound}

Define $ F( m_1, \ldots, m_r )$ to be the smallest size of a subset $S \subseteq \prod_i \Z_{m_i}$  such that:
$$ S+ \prod_i (\Z_{m_i}\setminus \{0\}) = \prod_i \Z_{m_i}.$$

Our proof in Section 3 actually shows that this quantity is a lower bound on the Fourier sparsity of delta functions on $\{0,1\}^r \subseteq \prod_{i} \Z_{m_i}$.
It is a common generalization of the bounds $r+1$ and $\prod_{i} \frac{m_i}{m_i-1}$.
Here we note a simple recursive inequality for it that is stronger than both the above bounds.
It implies, for example, that when $r = 2m$, that the Fourier sparsity of delta functions on $\Z_m^r$ is at least $\approx e \cdot m$ (while the best upper bound we know for this case is that there exist delta functions with Fourier sparsity $\leq m^2$).

\begin{lemma}
$$F(m_1, \ldots, m_r)  \geq \lceil \frac{m_r}{m_r -1}F(m_1, \ldots, m_{r-1})\rceil.$$ 
\end{lemma}
In particular, we have $$F(m_1, \ldots, m_r) \geq \max\left( F(m_1, \ldots, m_{r-1}) + 1,   \frac{m_r}{m_r - 1} F(m_1, \ldots, m_{r-1})\right),$$
and
$$ F(m,m, \ldots, m) \geq \begin{cases} r+1  & r \leq m-1 \\ m \cdot \left(\frac{m}{m-1}\right)^{r - m+1} &  r \geq m\end{cases}$$
\begin{proof}
Let $G = \prod_{i} \Z_{m_i}$.
Let $H = \prod_{i} ( \Z_{m_i} \setminus \{0\} ) $.
Suppose $S \subseteq G$ is such that $S + H = G$.

For $x \in \Z_m$, let $S_x$ be the set of all $v \in S$ with $v_r = x$. 
Choose $x$ with $|S_x|$ as large as possible: then $|S_x| \geq \frac{1}{m_r}|S|$.

Observe that for any $v \in S_x$ and any $w \in v + H$,
we have $w_r \neq x$. 

Thus:
$$\left( (S \setminus S_x)  +  H \right) \supseteq \left(\prod_{i < r} \Z_{m_i} \right) \times \{x \}.$$

Projecting to the first $r-1$ coordinates, we see that
the projection $S'$ of $S\setminus S_x$ to the first $r-1$ coordinates satisfies:
$$ S' +  \prod_{i < r} (\Z_{m_i} \setminus \{0\} ) = \prod_{i < r} \Z_{m_i}.$$

Thus $|S'| \geq F(m_1, \ldots, m_{r-1})$.
Combined with the inequality:
$$ |S'| \leq | S \setminus S_x | = |S|  - |S_x| \leq  \frac{m_r - 1}{m_r} |S|,$$
we get the result.
\end{proof}

In fact, the proof even shows that
$$F(m_1, \ldots, m_{r-1} ) \leq F(m_1, \ldots, m_r) - \left\lceil \frac{1}{m_r} F(m_1, \ldots, m_r) \right\rceil $$ (since in the notation of the proof, we get $|S'| \leq   |S| - \lceil \frac{1}{m_r} |S| \rceil$).

\section{Covering the cube by affine hyperplanes over $\F_p$}

A classical result of Alon and Furedi shows that any set $H$ of affine hyperplanes in $\R^r$ that cover all but one of the points of the Boolean hypercube must have size at least $r$. Specifically, if the hyperplanes of $H$ cover $\left(\{0,1\}^r \setminus \{0\} \right) \subseteq \R^r$ using affine hyperplanes, none of which contain $0$, then $|H| \geq r$.

Our results imply something nontrivial about the analogous covering question for $\left(\{0,1\}^r \setminus \{0\} \right) \subseteq \F_p^r$ for a prime $p$.
Specifically, we show that if $H$ is a set of hyperplanes, and:
$$H = H_1 \cup H_2 \ldots H_t$$
is the partition of $H$ into subsets of parallel hyperplanes,
then:
$$ \prod_{i} (|H_i| + 1 ) \geq \max\left( r+1, \left( \frac{p}{p-1} \right)^r\right).$$
On the other hand, there exist such sets $H_i$ with
$$ \prod_{i} (|H_i| + 1 ) \leq \left( p^{1/(p-1)} \right)^r$$

We now prove the lower bound.
If $W$ is a set of $w$ parallel hyperplanes, then there is a function $f_W : \F_p^r \to \mathbb C$ with Fourier sparsity $w+1$ and which vanishes exactly on the union of those hyperplanes. Taking the product of $f_{H_i}$ over all the $i$ and scaling by a constant, we get a delta function with Fourier sparsity $\prod_i (|H_i| + 1)$.
Then our Theorem~\ref{thm:bounds-same-domain} implies the result.

The upper bound follows by observing that the Fourier sparse function constructed in Claim~\ref{claim:upper-general} is of the above form.

\section{Delta functions on $\{0,1\}^2 \subseteq  \Z_{m_1} \times \Z_{m_2}$ over $\C$}
\label{sec:appendix-mobius}

In this appendix we consider the special case of delta functions on $\{0,1\}^2 \subseteq G = \Z_{m_1} \times \Z_{m_2}$ over the complex numbers.
We show that when $m_1$ and $m_2$ are coprime, there is no $\C$-valued
delta function on $\{0,1\}^2 \subseteq G$ with Fourier sparsity $\leq 3$. The argument uses
complex-analytic tools. 

By~\cite{Efremenko,IS,Cheeetal}, there are some $G$ of the form $\Z_{m_1} \times \Z_{m_2}$ (with $m_1, m_2$ coprime and odd) for which there exists a delta function $f: G \to \overline{\F}_2$ (the algebraic closure of $\F_2$) with Fourier sparsity $3$.
Thus our result needs to use some special property of complex numbers.

In brief, we express the property of being a delta function on $G$ in terms of vanishing of multilinear polynomials $P(X,Y)$ at roots of unity. Vanishing of the multilinear polynomial $P(X,Y)$  can be expressed in terms of a M\"{o}bius transformation $Y = A_P(X)$. Finally, we use properties of M\"{o}bius transformations on the unit disk in $\C$ to get the result.

\begin{proposition}\label{prop:complex-r2}
Let $m_1,m_2$ be coprime positive integers and let $G = \Z_{m_1}\times \Z_{m_2}$.
Suppose $f : G \to \C$ is a function that is delta on $\{0,1\}^2$, i.e.
\[
f(0,0) = 1,\qquad f(1,0)=f(0,1)=f(1,1)=0.
\]
Then the Fourier sparsity of $f$ is at least $4$.
\end{proposition}

\begin{proof}
Write $G = \Z_{m_1}\times \Z_{m_2}$ and fix primitive $m_1$-th and
$m_2$-th roots of unity $\omega_{m_1},\omega_{m_2}\in\C$.  Characters
of $G$ have the form
\[
\psi_{(a,b)}(x,y) \;=\; \omega_{m_1}^{a x}\,\omega_{m_2}^{b y},
\qquad (a,b)\in \Z_{m_1}\times \Z_{m_2}.
\]

Assume for contradiction that $f$ is a delta function on $\{0,1\}^2$
with Fourier sparsity at most $3$.  Assume that $f$ has a Fourier expansion
of the form
\[
f(x,y) \;=\; c_1 \psi_{(a_1,b_1)}(x,y) +
              c_2 \psi_{(a_2,b_2)}(x,y) +
              c_3 \psi_{(a_3,b_3)}(x,y),
\]
with $(a_i,b_i)$ pairwise distinct and $c_1,c_2,c_3\in\C$.
Since $f(0,0) =1$, not all the $c_i$ are zero.

Set
\[
\alpha_i := \omega_{m_1}^{a_i},\qquad \beta_i := \omega_{m_2}^{b_i}
\quad (i=1,2,3),
\]
so that each $\alpha_i$ is an $m_1$-th root of unity and each $\beta_i$
is an $m_2$-th root of unity.  The values of the three characters at the
four points $(0,0),(1,0),(0,1),(1,1)$ are:
\[
\psi_{(a_i,b_i)}(0,0) = 1,\quad
\psi_{(a_i,b_i)}(1,0) = \alpha_i,\quad
\psi_{(a_i,b_i)}(0,1) = \beta_i,\quad
\psi_{(a_i,b_i)}(1,1) = \alpha_i\beta_i.
\]
Thus, if we form the $4\times 3$ matrix
\[
M \;=\;
\begin{pmatrix}
1 & 1 & 1\\
\alpha_1 & \alpha_2 & \alpha_3\\
\beta_1 & \beta_2 & \beta_3\\
\alpha_1\beta_1 & \alpha_2\beta_2 & \alpha_3\beta_3
\end{pmatrix},
\]
the $i$-th column of $M$ is the column vector of values of
$\psi_{(a_i,b_i)}$ on $\{0,1\}^2$.

The vector of values of $f$ on $\{0,1\}^2$ is then
\[
M c \;=\; 
\begin{pmatrix}
1\\ 0\\ 0\\ 0
\end{pmatrix},
\qquad
\mbox{ where }
c := \begin{pmatrix}
c_1\\ c_2\\ c_3
\end{pmatrix}
\]
In other words, the existence of a $3$-sparse delta function with these
three characters is equivalent to the statement that
\[
e_1 := \begin{pmatrix}
1\\ 0\\ 0\\ 0
\end{pmatrix}
\]
lies in the column space of $M$.

By elementary linear algebra, $e_1$ lies in the column space of $M$ if
and only if every row vector $v\in\C^4$ that is orthogonal to all the
columns of $M$ is also orthogonal to $e_1$, i.e.
\[
v \cdot M = 0 \;\Rightarrow\; v \cdot e_1 = 0.
\]
We now interpret such row vectors $v$ in terms of multilinear
polynomials evaluated at the $(\alpha_i, \beta_i)$.  Write
\[
v = (v_1,v_X,v_Y,v_{XY}) \in \C^4
\]
and associate to $v$ the multilinear polynomial
\[
P_v(X,Y) \;=\; v_1 + v_X X + v_Y Y + v_{XY} X Y.
\]
Observe that for $i=1,2,3$ we have
\[
v \cdot \text{(column $i$ of $M$)}
\;=\; P_v(\alpha_i,\beta_i).
\]
Thus $v \cdot M = 0$ is equivalent to
\[
P_v(\alpha_1,\beta_1) =
P_v(\alpha_2,\beta_2) =
P_v(\alpha_3,\beta_3) = 0.
\]
On the other hand
\[
v \cdot e_1 = v_1 = P_v(0,0).
\]

We conclude that the existence of a $(\leq 3)$-sparse delta function with
support $\{\psi_{(a_i,b_i)}\}_{i=1}^3$ is equivalent to the following
statement:

\medskip\noindent
\emph{Every multilinear polynomial $P(X,Y)$ that vanishes at the three
points $(\alpha_i,\beta_i)$, $i=1,2,3$, must also vanish at $(0,0)$.}
\medskip

We now show that this statement is false.
\begin{lemma}
\label{lem:multi}
Let $m_1, m_2$ be relatively prime integers.
Let $\alpha_1, \alpha_2, \alpha_3 \in \mathbb C$ be $m_1$-th roots of $1$,
and $\beta_1, \beta_2, \beta_3 \in \mathbb C$ be $m_2$-th roots of $1$.

Then there exists a multilinear polynomial $P(X,Y)$ such that:
$P(\alpha_i,\beta_i) = 0$ for each $i$, and $P(0,0) \neq 0$.
\end{lemma}
\begin{proof}
    We take cases on whether the $\alpha_i$ and $\beta_i$ are all distinct or not.

    \begin{itemize}
        \item {\bf Case 1:} Some two (or more) of the $\alpha_i$ are equal. Without loss of generality, assume $\alpha_1 = \alpha_2$. Then $P(X,Y) = (X-\alpha_1)(Y-\beta_3)$ is a multilinear polynomial that vanishes at all the $(\alpha_i, \beta_i)$ and is nonzero at $(0,0)$.
        \item {\bf Case 2:} Some two (or more) of the $\beta_i$ are equal. Without loss of generality, assume $\beta_1 = \beta_2$. Then $P(X,Y) = (X-\alpha_3)(Y-\beta_1)$ is a multilinear polynomial that vanishes at all the $(\alpha_i, \beta_i)$ and is nonzero at $(0,0)$.
        \item {\bf Case 3:} All the $\alpha_i$ are distinct and all the $\beta_i$ are distinct. In this case, we use the basic fact that M\"{o}bius transformations can take any $3$ distinct points to any $3$ distinct points.
        Thus, there is a M\"{o}bius transformation
        $$A(X) = - \frac{v_1 + v_X X}{v_Y + v_{XY} X }$$
        with $A(\alpha_i) = \beta_i$ for each $i$.
        Rearranging this, we get that for the multilinear polynomial $P(X,Y)$ given by:
        \begin{align}
            \label{eq:PfromMob}
            P(X,Y) = v_1 + v_X X + v_Y Y + v_{XY} XY
        \end{align}
        satisfies, for each $i$:
        $$ P(\alpha_i, \beta_i) = 0.$$

        We now prove that $P(0,0)$ is nonzero.
        This is equivalent to showing that $A(0) \neq 0$. 
        Suppose not; namely suppose $A(0) = 0$.
        Then the map $A : \mathbb C \to \mathbb C$
        satsifies:
        \begin{itemize}
            \item $A(\alpha_i) = \beta_i$ for each $i$. In particular, since all $\alpha_i$ are distinct, all $\beta_i$ are distinct, they lie on the unit circle $|z| = 1$, and M\"{o}bius transformations take circles to circles, we get that $A$ maps the unit circle to itself.

            \item $A(0) = 0$.
            \item Thus $A$ maps the open unit disc $\D = \{ z \mid |z|  < 1\}$ to itself.
        \end{itemize}

        It is well known~\cite[Chapter 12]{RudinRealComplexAnalysis} that M\"{o}bius transformations mapping $\D$ to $\D$ are of the form:
        $$A(z) =  \frac{a z - \bar{b}}{b z - \bar{a}}.$$
        If in addition we have $A(0) =0$, then we get:
        $$0 = A(0) = \frac{-\bar{b}}{-\bar{a}},$$
        which means that $b = 0$.
        Thus $A(z)$ is of the form
        $$ A(z) = \lambda z,$$ 
        for some $\lambda = \frac{a}{-\bar{a}}\in \mathbb C$.

Now $A(\alpha_i) = \beta_i$ for $i=1,2,3$ implies $\beta_i = \lambda
\alpha_i$.  In particular, for any $i\neq j$,
\[
\frac{\beta_i}{\beta_j}
= \frac{\lambda\alpha_i}{\lambda\alpha_j}
= \frac{\alpha_i}{\alpha_j}.
\]
The left-hand side is an $m_2$-th root of unity, and the right-hand side
is an $m_1$-th root of unity.  Since $\gcd(m_1,m_2)=1$, the only complex
number that is simultaneously an $m_1$-th and an $m_2$-th root of unity
is $1$.
So we must have $\beta_i/\beta_j = 1$ and hence $\beta_i =
\beta_j$, contradicting the fact that we are in the case where the $\beta_i$ are all distinct. 

Thus our assumption that $P(0,0) = 0$ is wrong, and we conclude that the multilinear polynomial $P(X,Y)$ specified by Equation~\eqref{eq:PfromMob} is as desired.
    \end{itemize}
This concludes the proof of Lemma~\ref{lem:multi}.
\end{proof}

By the discussion preceding Lemma~\ref{lem:multi}, we conclude that there are no delta functions on $G$ with Fourier sparsity $\leq 3$.
\end{proof}

\end{document}